\documentclass[preprint,12pt]{larticle}
\usepackage[a4paper, total={6.18in, 10in}]{geometry}
\usepackage[utf8]{inputenc}

\usepackage{amssymb}
\usepackage{amsthm}
\usepackage{amsmath}
\usepackage{tikz}

\usepackage[symbol]{footmisc}

\usepackage{xcolor}

\newcommand{\cS}[1]{\mathcal{S}_{#1}}
\newcommand{\RR}{\mathbb{R}}

\newtheorem{lemma}{Lemma}
\newtheorem{theorem}{Theorem}

\begin{document}

\begin{frontmatter}

\title{Polynomial encoding of rooted trees with branch lengths}

\author{Pengyu Liu{\footnotesize$^{1,2}$\footnote{To whom correspondence should be addressed; e-mail: pengyu.liu@uri.edu.}}, Joan Carles Pons{\footnotesize{$^3$}}, Gerard Ribas{\footnotesize{$^4$}}}

\address{\footnotesize{$^1$}Department of Mathematics and Applied Mathematical Sciences and \footnotesize{$^2$}Department of Cell and Molecular Biology, University of Rhode Island, Kingston, RI 02881, USA}

\address{\footnotesize{$^3$}Department of Mathematics and Computer Science, University of the Balearic Islands, 07122, Palma, Spain}

\address{\footnotesize{$^4$}Higher Polytechnic School, University of the Balearic Islands, 07122, Palma, Spain}

\begin{abstract}
Phylogenetic trees are rooted trees with branch lengths that record genetic divergence or elapsed time, and quantifying differences between them is central to a wide range of evolutionary and epidemiological analyses. 
Graph-polynomial encodings of rooted trees provide an accurate, interpretable, and computationally efficient way to compare tree shapes, but existing polynomial encodings must be paired with auxiliary structures 
to study rooted trees with branch lengths.
We introduce a bivariate polynomial encoding that incorporates branch lengths directly into a recursive computation from the leaf vertices to the root vertex of a tree. 
We prove that, for rooted trees with branch lengths and no vertices of degree two, which include all standard phylogenetic trees, two trees have the same polynomial if and only if their underlying unlabeled trees are isomorphic and the branch lengths of corresponding edges are equal. 
We apply the polynomial encoding to three published HIV-1 phylogenies sampled in different epidemiological settings and show that it accurately separates the three datasets based on their tree topologies and branch lengths, outperforming previous polynomial-based approaches for analyzing rooted trees with branch lengths.
\end{abstract}

\end{frontmatter}


\section{Introduction}
\label{S:1}

Rooted trees with branch lengths are central to phylogenetic analysis. 
The leaf vertices of a phylogenetic tree represent the sampled taxa, for example, viral genetic sequences from different infected individuals, or specimens from different species.
The internal vertices represent the inferred common ancestors of the sampled taxa.
Each edge or branch of a phylogenetic tree is labeled with a length that records either the amount of genetic change accumulated along the branch, typically measured in expected nucleotide substitutions per site, or, in time-calibrated phylogenies, the elapsed time between two branching events \cite{felsenstein2004, bouckaert2019}. 
Trees of both types are routinely reconstructed from molecular sequence data and form the basis for downstream inference of evolutionary rates, ancestral states, divergence times, and transmission dynamics in viral epidemics \cite{volz2013}.

An important task in such analyses is to quantify how similar or different two trees are. 
This has motivated a range of tree representations and tree-comparison metrics. 
One recent approach is to encode each tree as a graph polynomial which can be computed recursively from the leaf vertices to the root vertex of a tree, with coefficients of the polynomial summarizing the structural features of the tree \cite{liu2021tree}.
Such polynomial encodings are accurate, in the sense that two trees produce distinct polynomials whenever they are not isomorphic, and therefore no structural information is lost. 
They are interpretable, since each coefficient counts a specific class of subtrees.
The polynomial encodings are also computationally efficient \cite{liu2021tree,liu2022metric}.
These polynomial encodings have been applied to phylogenetic analysis, extracting evolutionary information from different epidemiological or evolutionary processes \cite{liu2022metric}, and to RNA secondary structure analysis, helping identify features of nascent RNA folding that associate with the formation of R-loops, a class of non-canonical nucleic acid structures linked to human diseases including cancer and genetic disorders \cite{liu2024rloop}.

In phylogenetic analysis, the branch lengths of a tree typically carry as much information as the topology -- they directly measure divergence or time, and trees with identical topology but different branch lengths describe different evolutionary scenarios. 
The polynomial encodings developed so far have to pair with a tree lattice for analyzing rooted trees with branch lengths \cite{liu2022metric}. 
Here, we introduce a polynomial encoding for rooted trees that incorporates branch lengths directly into the recursive computation, without the need for additional structures such as the tree lattice. 
We focus on the class of rooted trees with branch lengths that have no vertices of degree two, which include all standard phylogenetic trees, and we call such a rooted tree with branch lengths a {\em branching tree}.
We show that two branching trees have the same polynomial if and only if their underlying unlabeled rooted trees are isomorphic and the corresponding edges have equal branch lengths. 
We then apply the polynomial encoding, together with the Canberra distance \cite{liu2022metric,liu2026comp} between the coefficients, to phylogenetic trees reconstructed from HIV-1 sequences sampled in three different epidemiological settings.
We show that k-medoids clustering \cite{kaufman1990} on the polynomial-based Canberra distance recovers the three source trees at 100\% accuracy when the original branch lengths are retained, and at 99.3\% when the branch lengths are normalized to a common scale.
This polynomial encoding for rooted trees with branch lengths outperforms the previously studied method that involves polynomial encodings and a tree lattice on the same data \cite{liu2022metric}.

\section{Polynomial encoding of rooted trees}
\label{S:2}

\subsection{Background}

A {\em tree} $T$ is a graph without any cycles. 
A tree $T$ is {\em rooted} if a single vertex in $T$ is identified as the {\em root vertex} of tree. 
A vertex in a rooted tree $T$ is a {\em leaf vertex} if and only if the vertex is not the root vertex and has degree 1.
A vertex in a rooted tree $T$ is an {\em elementary vertex} if and only if the vertex is not the root vertex and has degree 2.

Throughout this paper, rooted trees are considered {\em unlabeled}, meaning
that vertices and edges carry no identifying labels unless explicitly stated otherwise. The distinguished root, and the classification of vertices as leaves, internal
vertices, or elementary vertices, are part of the rooted-tree structure and
are not regarded as labels.


A {\em rooted tree with branch lengths} $B$ is a rooted tree such that each edge of $B$ is labeled with a number in $\RR_{>0}$. These lengths are weights to be preserved by isomorphisms, not identifying labels for the edges.

A bivariate polynomial $L(T,x,y)$ is defined for an unlabeled rooted tree $T$ in \cite{liu2021tree}.
The polynomial is defined recursively from leaf vertices to the root vertex of the tree $T$ as follows. 
If a vertex $v$ of $T$ is a leaf vertex, then the polynomial at $v$ is defined to be $L(v,x,y)= x$.
If $v$ is an internal vertex with $k$ child vertices, $\{v_1, v_2, \ldots, v_k\}$, where $k>0$, then, the polynomial at $v$ is given by formula (\ref{polyl}).
\begin{equation}\label{polyl}
    L(v,x,y) = y + \prod_{i = 1}^{k} L(v_i,x,y)
\end{equation}
The polynomial at the root vertex of $T$ is the polynomial for the rooted tree $T$.
It is proved in \cite{liu2021tree} that the polynomial is a complete invariant for unlabeled rooted trees and a generating function for a class of subtrees called primary subtrees. 


\subsection{A polynomial invariant for rooted trees with branch lengths}

To define a polynomial for rooted trees with branch lengths, we first modify the definition of the polynomial $L(T,x,y)$ and define another polynomial $P(T,x,y)$ for an unlabeled rooted tree $T$.
The polynomial $P(T,x,y)$ is also defined recursively from the leaf vertices to the root vertex of $T$. 
If a vertex $v$ of $T$ is a leaf vertex, then the polynomial at $v$ is defined to be $P(v,x,y)= x$.
If $v$ is an internal vertex with $k$ child vertices, $\{v_1, v_2, \ldots, v_k\}$, where $k>0$, then, the polynomial at $v$ is given by formula (\ref{polyp}).
\begin{equation}\label{polyp}
    P(v,x,y) = \prod_{i = 1}^{k} \Bigl(y + P(v_i,x,y)\Bigr)
\end{equation}
The polynomial at the root vertex of $T$ is the polynomial for the rooted tree $T$, and is denoted by $P(T,x,y)$.

The polynomial $P(T,x,y)$ is defined by traversing the edges from leaf vertices to the root vertex. See Figure~\ref{fig1}A for an example.
Whenever an edge is traversed, the variable $y$ is added to the polynomial. 
This allows us to include branch length information to each edge in the tree $T$.

\begin{figure}[ht!]
    \centering
    \begin{tikzpicture}[
  vert/.style       = {circle, fill=black, inner sep=1.5pt},
  edgelabel/.style  = {font=\small\itshape, inner sep=2pt},
  polylabel/.style  = {font=\small},
  panellabel/.style = {font=\small\bf},
  caption/.style    = {font=\small},
]

\begin{scope}[shift={(0,0)}]
  \node[panellabel] at (-6.5, 3) {A};
  \node[vert,
        label={[polylabel] above:{$(x+y)(x+y)\bigl[(x+y)(x+y)+y\bigr]$}}]
       (r) at (0, 2.5) {};
  \node[vert, label={[polylabel] below:{$x$}}]            (l1) at (-4.5,  0.0) {};
  \node[vert, label={[polylabel] below:{$x$}}]            (l2) at (-1.0,  0.0) {};
  \node[vert, label={[polylabel] right:{$(x+y)(x+y)$}}]   (m)  at ( 4.0,  0.0) {};
  \node[vert, label={[polylabel] below:{$x$}}]            (l3) at ( 3.0, -2.0) {};
  \node[vert, label={[polylabel] below:{$x$}}]            (l4) at ( 5.0, -2.0) {};
  \draw (r) -- (l1);
  \draw (r) -- (l2);
  \draw (r) -- (m);
  \draw (m) -- (l3);
  \draw (m) -- (l4);
  \node[edgelabel] at (-2.7,  1.2) {$y$};
  \node[edgelabel] at (-0.7,  1.2) {$y$};
  \node[edgelabel] at ( 1.55,  1.2) {$y$};
  \node[edgelabel] at ( 3.3, -1) {$y$};
  \node[edgelabel] at ( 4.75, -1) {$y$};
  \node[caption] at (-2.4, -2.1) {unlabeled rooted tree $T$};
  \node[caption, anchor=west] at (-5.1, -3.25)
       {$P(T,\,x,\,y) = x^{4} + 4x^{3}y + 6x^{2}y^{2} + 4xy^{3} + y^{4} + x^{2}y + 2xy^{2} + y^{3}$};
\end{scope}

\begin{scope}[shift={(0,-8.0)}]
  \node[panellabel] at (-6.5, 3) {B};
  \node[vert,
        label={[polylabel] above:{$(x+3y)(x+y)\bigl[(x+2y)(x+y)+0.5y\bigr]$}}]
       (r) at (0, 2.5) {};
  \node[vert, label={[polylabel] below:{$x$}}]            (l1) at (-4.5,  0.0) {};
  \node[vert, label={[polylabel] below:{$x$}}]            (l2) at (-1.0,  0.0) {};
  \node[vert, label={[polylabel] right:{$(x+2y)(x+y)$}}]  (m)  at ( 4.0,  0.0) {};
  \node[vert, label={[polylabel] below:{$x$}}]            (l3) at ( 3.0, -2.0) {};
  \node[vert, label={[polylabel] below:{$x$}}]            (l4) at ( 5.0, -2.0) {};
  \draw (r) -- (l1);
  \draw (r) -- (l2);
  \draw (r) -- (m);
  \draw (m) -- (l3);
  \draw (m) -- (l4);
  \node[edgelabel] at (-2.7,  1.2) {$3$};
  \node[edgelabel] at (-0.7,  1.2) {$1$};
  \node[edgelabel] at ( 1.5,  1.2) {$0.5$};
  \node[edgelabel] at ( 3.3, -1) {$2$};
  \node[edgelabel] at ( 4.75, -1) {$1$};
  \node[caption] at (-2.1, -2.1) {rooted tree with branch lengths $B$};
  \node[caption, anchor=west] at (-5.75, -3.25)
       {$P(B,\,x,\,y) = x^{4} + 7x^{3}y + 17x^{2}y^{2} + 17xy^{3} + 6y^{4} + 0.5x^{2}y + 2xy^{2} + 1.5y^{3}$};
\end{scope}

\end{tikzpicture}
    \caption{{\bf Computation of the polynomial.} {A:} An illustration of the recursive computation of the polynomial $P(T,x,y)$ for an unlabeled rooted tree $T$. {B:} An illustration of the recursive computation of the polynomial $P(B,x,y)$ for a rooted tree with branch lengths $B$. Numbers on edges indicate the corresponding branch lengths.}
    \label{fig1}
\end{figure}
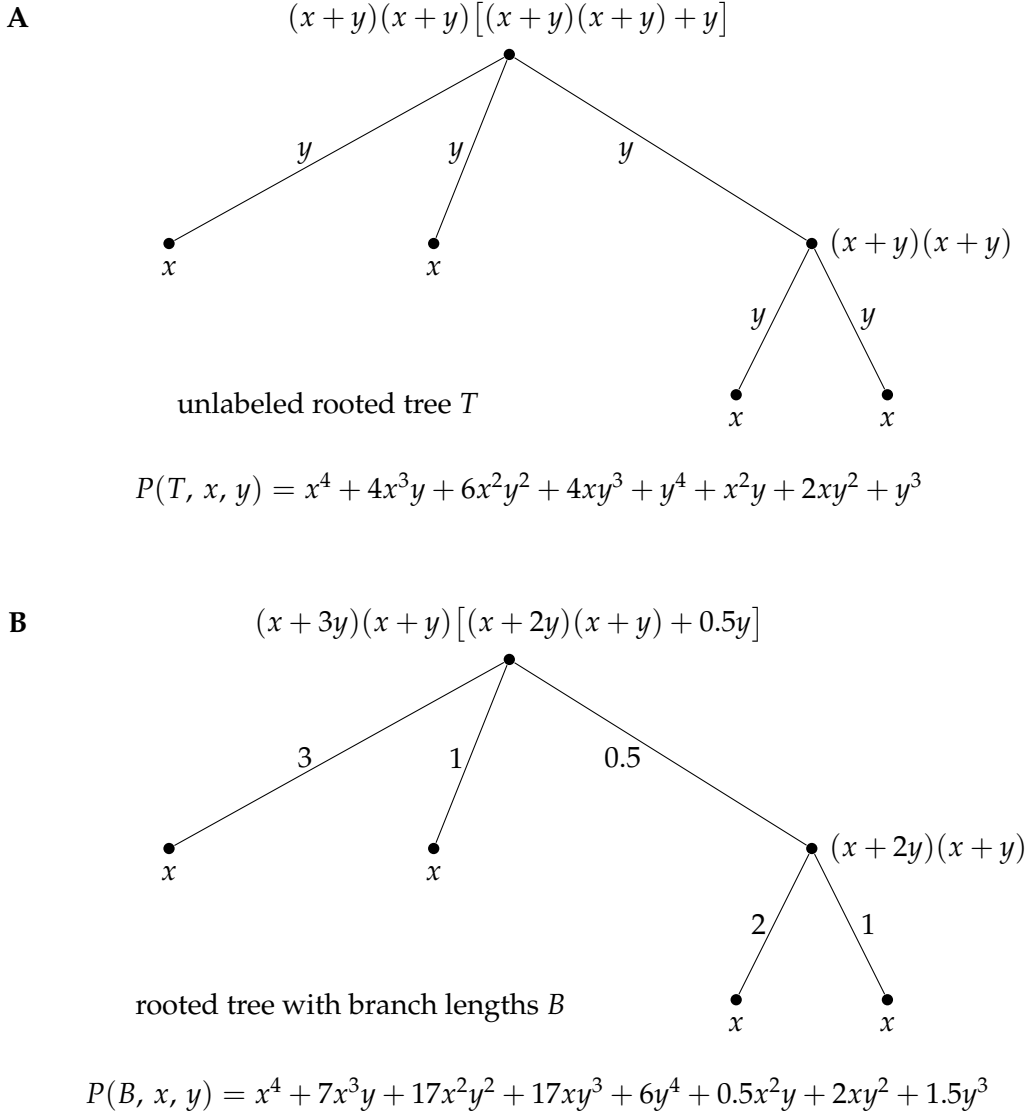


\begin{theorem}
    Two unlabeled rooted trees $T_1$ and $T_2$ are isomorphic if and only if the two polynomials $P(T_1,x,y)$ and $P(T_2,x,y)$ are identical.
\end{theorem}

\begin{proof}
    It is trivial to check that the following equation (\ref{variablechange}) holds for any unlabeled rooted tree $T$. 
    \begin{equation}\label{variablechange}
    P(T,x,y) = L(T, x+y, y) - y
    \end{equation}
    So, the polynomial $P(T,x,y)$ being a complete invariant for unlabeled rooted trees directly follows that the polynomial $L(T,x,y)$ being a complete invariant for unlabeled rooted trees \cite{liu2021tree}. 
\end{proof}

Here, we generalize the polynomial to encode rooted trees with branch lengths including the branch length of each edge in a tree.
To distinguish from unlabeled rooted trees, we use $B$ to denote a rooted tree with branch lengths. 
For a rooted tree with branch lengths $B$, we define the polynomial $P(B,x,y)$ recursively from the leaf vertices of $B$ to the root vertex of $B$ as follows.
If a vertex $v$ of $B$ is a leaf vertex, then the polynomial at $v$ is defined to be $P(v,x,y)= x$.
If $v$ is an internal vertex with $k$ child vertices, $\{v_1, v_2,\ldots, v_k\}$, with branch length $w_i$ on the edge $e_i=vv_i$, then, the polynomial at $v$ is given by formula (\ref{polypw}).
\begin{equation}\label{polypw}
    P(v,x,y) = \prod_{i = 1}^{k} \Bigl(w_iy + P(v_i,x,y)\Bigr)
\end{equation}
The polynomial at the root vertex of $B$ is the polynomial for the rooted tree $B$. 
See Figure~\ref{fig1}B for an example.
Without ambiguity, we denote $P(B,x,y)$ by $P(B)$.


Two rooted trees with branch lengths $B_1$ and $B_2$ are {\em isomorphic} if and only if the underlying unlabeled rooted trees of $B_1$ and $B_2$ are isomorphic and the corresponding edges in $B_1$ and $B_2$ have equal branch lengths. 
It is trivial that the polynomial $P(B,x,y)$ is an invariant of rooted trees with branch lengths. 
\begin{theorem}\label{onlyif}
    If two rooted trees with branch lengths $B_1$ and $B_2$ are isomorphic, then the two polynomials $P(B_1,x,y)$ and $P(B_2,x,y)$ are identical.
\end{theorem}






\section{Interpretation of the polynomial}%
\label{sec:interpretation}



Given $k>0$ rooted trees with branch lengths $B_1, B_2, \ldots,B_k$ with roots $r_1, r_2,\ldots,r_k$, respectively, the {\em wedge product} of the $k$ trees is a rooted tree with branch lengths $B$ obtained by connecting the $k$ roots $r_1, r_2,\ldots,r_k$ to a common root $r$ of $B$ with $k$ edges $rr_1, rr_2,\ldots,rr_k$, respectively.
Let $w_1,w_2,\ldots,w_k$ be the branch lengths of the $k$ edges, respectively.
We denote the wedge product by $B = \wedge_{i=1}^k w_iB_i$.
Note that for any rooted tree with branch lengths $B$ with at least an edge, $B$ is of form $B = \wedge_{i=1}^k w_iB_i$.
Moreover, $P(B,x,y)=\prod_{i=1}^k(w_iy+p(B_i))$.



Let $B$ be a rooted tree with branch lengths. 
A subtree $S$ of $B$ is a {\em principal subtree} of $B$ if and only if the root vertex of $S$ is the root vertex of $B$.
{\em Pruning} an edge $e$ from $B$ means removing the edge $e$ and all its descendant vertices and edges from $B$.
Two edges $e_1$ and $e_2$ of $B$ are {\em incomparable} if and only if $e_1$ is not a descendant of $e_2$ and $e_2$ is not a descendant of $e_1$.
It is trivial that a principal subtree $S$ of $B$ can be obtained by pruning a set of pairwise incomparable edges from $B$. See Figure~\ref{fig2} for an example.

We denote $\cS{B}$ the set of principal subtrees of $B$.
Let $S\in\cS{B}$ be a principal subtree obtained by pruning the set of pairwise incomparable edges $\{e_1,e_2,\ldots,e_\beta \}$ of $B$. 
We assign to $S$ the monomial $q(S)=wx^\alpha y^\beta$, where $\alpha$ is the number of leaf vertices of $B$ that are in $S$, and $w$ is the product of the branch lengths of $e_1,e_2\ldots,e_\beta$.
See Figure~\ref{fig2} for examples. 

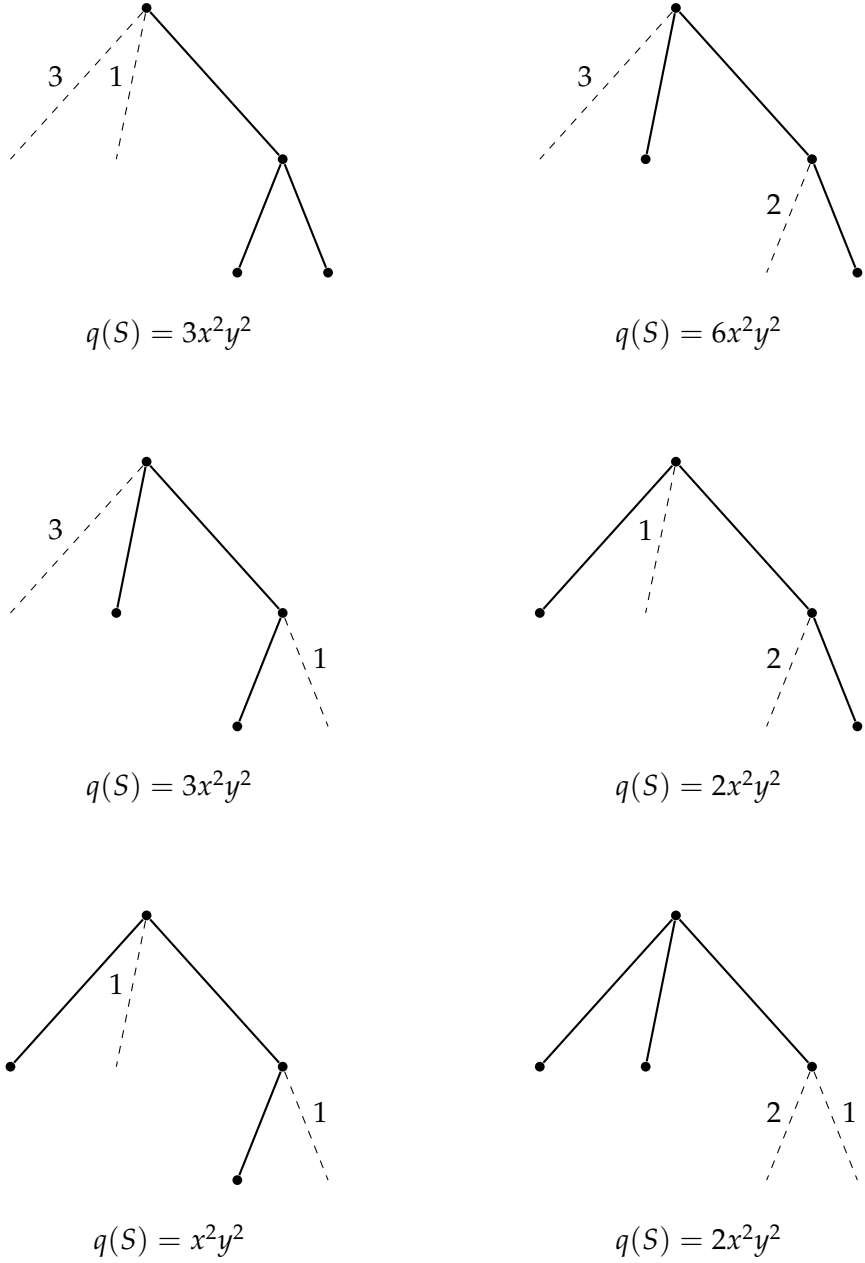
\begin{figure}[ht!]
    \centering
\begin{tikzpicture}[
  vert/.style      = {circle, fill=black, inner sep=1.3pt},
  edgelabel/.style = {font=\small\itshape, inner sep=2pt},
  caption/.style   = {font=\small},
  solidedge/.style = {thick},
  dashedge/.style  = {dashed},
]
\begin{scope}[shift={(0,0)}]
  \node[vert]  (r)  at (0,2)      {};
  \coordinate  (l1) at (-1.8,0);
  \coordinate  (l2) at (-0.4,0);
  \node[vert]  (m)  at (1.8,0)    {};
  \node[vert]  (l3) at (1.2,-1.5) {};
  \node[vert]  (l4) at (2.4,-1.5) {};
  \draw[dashedge] (r) -- (l1);
  \draw[dashedge] (r) -- (l2);
  \draw[solidedge] (r) -- (m);
  \draw[solidedge] (m) -- (l3);
  \draw[solidedge] (m) -- (l4);
  \node[edgelabel] at (-1.2, 1.1) {$3$};
  \node[edgelabel] at (-0.4, 1.1) {$1$};
  \node[caption] at (0.3,-2.3) {$q(S) = 3x^2y^2$};
\end{scope}
\begin{scope}[shift={(7,0)}]
  \node[vert]  (r)  at (0,2)      {};
  \coordinate  (l1) at (-1.8,0);
  \node[vert]  (l2) at (-0.4,0)   {};
  \node[vert]  (m)  at (1.8,0)    {};
  \coordinate  (l3) at (1.2,-1.5);
  \node[vert]  (l4) at (2.4,-1.5) {};
  \draw[dashedge] (r) -- (l1);
  \draw[solidedge] (r) -- (l2);
  \draw[solidedge] (r) -- (m);
  \draw[dashedge] (m) -- (l3);
  \draw[solidedge] (m) -- (l4);
  \node[edgelabel] at (-1.2, 1.1) {$3$};
  \node[edgelabel] at ( 1.3,-0.6) {$2$};
  \node[caption] at (0.3,-2.3) {$q(S) = 6x^2y^2$};
\end{scope}
\begin{scope}[shift={(0,-6)}]
  \node[vert]  (r)  at (0,2)      {};
  \coordinate  (l1) at (-1.8,0);
  \node[vert]  (l2) at (-0.4,0)   {};
  \node[vert]  (m)  at (1.8,0)    {};
  \node[vert]  (l3) at (1.2,-1.5) {};
  \coordinate  (l4) at (2.4,-1.5);
  \draw[dashedge] (r) -- (l1);
  \draw[solidedge] (r) -- (l2);
  \draw[solidedge] (r) -- (m);
  \draw[solidedge] (m) -- (l3);
  \draw[dashedge] (m) -- (l4);
  \node[edgelabel] at (-1.2, 1.1) {$3$};
  \node[edgelabel] at ( 2.3,-0.6) {$1$};
  \node[caption] at (0.3,-2.3) {$q(S) = 3x^2y^2$};
\end{scope}
\begin{scope}[shift={(7,-6)}]
  \node[vert]  (r)  at (0,2)      {};
  \node[vert]  (l1) at (-1.8,0)   {};
  \coordinate  (l2) at (-0.4,0);
  \node[vert]  (m)  at (1.8,0)    {};
  \coordinate  (l3) at (1.2,-1.5);
  \node[vert]  (l4) at (2.4,-1.5) {};
  \draw[solidedge] (r) -- (l1);
  \draw[dashedge] (r) -- (l2);
  \draw[solidedge] (r) -- (m);
  \draw[dashedge] (m) -- (l3);
  \draw[solidedge] (m) -- (l4);
  \node[edgelabel] at (-0.4, 1.1) {$1$};
  \node[edgelabel] at ( 1.3,-0.6) {$2$};
  \node[caption] at (0.3,-2.3) {$q(S) = 2x^2y^2$};
\end{scope}
\begin{scope}[shift={(0,-12)}]
  \node[vert]  (r)  at (0,2)      {};
  \node[vert]  (l1) at (-1.8,0)   {};
  \coordinate  (l2) at (-0.4,0);
  \node[vert]  (m)  at (1.8,0)    {};
  \node[vert]  (l3) at (1.2,-1.5) {};
  \coordinate  (l4) at (2.4,-1.5);
  \draw[solidedge] (r) -- (l1);
  \draw[dashedge] (r) -- (l2);
  \draw[solidedge] (r) -- (m);
  \draw[solidedge] (m) -- (l3);
  \draw[dashedge] (m) -- (l4);
  \node[edgelabel] at (-0.4, 1.1) {$1$};
  \node[edgelabel] at ( 2.3,-0.6) {$1$};
  \node[caption] at (0.3,-2.3) {$q(S) = x^2y^2$};
\end{scope}
\begin{scope}[shift={(7,-12)}]
  \node[vert]  (r)  at (0,2)      {};
  \node[vert]  (l1) at (-1.8,0)   {};
  \node[vert]  (l2) at (-0.4,0)   {};
  \node[vert]  (m)  at (1.8,0)    {};
  \coordinate  (l3) at (1.2,-1.5);
  \coordinate  (l4) at (2.4,-1.5);
  \draw[solidedge] (r) -- (l1);
  \draw[solidedge] (r) -- (l2);
  \draw[solidedge] (r) -- (m);
  \draw[dashedge] (m) -- (l3);
  \draw[dashedge] (m) -- (l4);
  \node[edgelabel] at ( 1.3,-0.6) {$2$};
  \node[edgelabel] at ( 2.3,-0.6) {$1$};
  \node[caption] at (0.3,-2.3) {$q(S) = 2x^2y^2$};
\end{scope}
\end{tikzpicture}
    \caption{{\bf Principal subtrees and monomials.} Principal subtrees of the rooted tree with branch lengths $B$ in Figure~\ref{fig1}B associated with the term $17x^2y^2$ in $P(B,x,y)$. Edges with dashed lines are not in the principal subtrees. They indicate edges pruned to construct the principal subtrees. The coefficient for each monomial is the product of the product of pruned edges' branch lengths. All six coefficient sum up to the coefficient $17$ of the term $17x^2y^2$ in $P(B,x,y)$.}
    \label{fig2}
\end{figure}

\begin{theorem}
  For a rooted tree with branch lengths $B$ be, $P(B,x,y)=\sum_{S\in\cS{B}}q(S)$.
\end{theorem}



\begin{proof}
  We proceed by strong induction on $d$, the depth of $B$.

  If $d=0$, then $B$ consists on a single vertex, so $P(B,x,y)=x$. Note that in this
  case $\cS{B}=\{B\}$ and $q(B)=x$, hence $P(B,x,y)=\sum_{S\in\cS{B}}q(S)$.

  Assume that $P(B,x,y)=\sum_{S\in\cS{B}}q(S)$ for any rooted tree with brach lengths $B$ that has depth $d\le D$. 
  Suppose that $B'$ has depth $d=D+1$, and let $B_1,B_2,\ldots,B_k$ be trees such
  that $B'=\wedge_{i=1}^kw_iB_i$, where $k>0$ and $w_1,w_2,\ldots,w_n\in \RR_{>0}$. 

  Notice that the depth of the trees $B_i$ is at
  most $D$. Thus, by induction hypothesis, we have the following formula (\ref{formqs}).
  \begin{align}\label{formqs}
      P(B')=\prod_{i=1}^k (w_iy+P(B_i)) =\prod_{i=1}^k\Bigl(w_iy+\sum_{S\in\cS{B_i}}q(S)\Bigr)
  \end{align}
  Note that for any principal subtree $S'$ of $B'$, $S'$ is a wedge product $S'=\wedge_{i\in A} S_i$ for some $A\subseteq [k] = \{1,2,\ldots,k\}$ and some $S_i\in \cS{B_i}$ for every $i \in A$. 
  The set $A$ can be empty, and in this case, $S'$ is the principal subtree comprising only the root vertex of $B'$.
  The distributive law implies the following formula (\ref{distp}).
  \begin{align}\label{distp}
P(B')=\prod_{i=1}^k\Bigl(w_iy+\sum_{S\in\cS{B_i}}q(S)\Bigr) = 
      \sum_{A\subseteq[k]} \Bigl(\prod_{i\in [k]\setminus A}w_iy\Bigr) \Bigl(\prod_{i\in A}\sum_{S_i\in\cS{B_i}}q(S_i)\Bigr)
  \end{align}

We claim that there is a one-to-one correspondence between the set of monomials $\{q(S') \, | \, S'\in \cS{B'}\}$ and the set of monomials in the above formula (\ref{distp}), which have the form $(\prod_{i\in [k]\setminus A}w_iy) (\prod_{i\in A}\sum_{S_i\in\cS{B_i}}q(S_i))$. 
This is straightforward as any principal subtree $S'=\wedge_{i\in A} S_i$ of $B'$ is constructed by pruning the edges in each $S_i$ as well as the edge connecting the root of $B'$ and the root of $B_i$ for each $i\in [k]\setminus A$.  

Therefore, for any rooted tree with branch lengths, $P(B,x,y)=\sum_{S\in\cS{B}}q(S)$.
\end{proof}

Let $B$ be a rooted tree with branch lengths that has $n$ leaf vertices. 
Notice that $B\in\cS{B}$, $q(B)=x^n$, and $B$
is the unique principal subtree of $B$ that has the $n$ leaf vertices of $B$. 
Then, the leading term of $P(B)$ in $(\RR[y])[x]$ is $x^n$. 
Moreover, $y$ divides
every other term in $P(B)$, since all the trees in $\cS{B}\setminus\{B\}$ are
obtained by pruning at least one edge of $B$.

\section{Complete polynomial invariant for branching trees}%
\label{sec:invariant}

\begin{lemma}\label{lem:irred}
  For any rooted tree with branch lengths $B$, $wy+P(B,x,y)$ is irreducible in $\RR[x,y]$ for
  any $w\in\RR_{>0}$.
\end{lemma}

\begin{proof}
  We will use the Eisenstein criterion to prove the result. 
  Let $n$ be the number of leaf vertices of $B$, and $I=\langle y\rangle$ be a (prime) ideal in $\RR[y]$. 
  From the previous section, the leading of $wy+P(B)$ in $(\RR[y])[x]$ is $x^n$, with
  coefficient 1, and $1\notin I$. Moreover, $y$ divides any other term
  of $wy+P(B)$, so the coefficient of every term of $wy+P(B)-x^n$ is in $I$.
  Also, as $w>0$ and $P(B)$ contains a term of the form $w'y$ for
  some $w'\ge 0$, we have that the independent term of $P(B)$ in $(\RR[y])[x]$
  contains $(w+w')y$ with $w+w'\in\RR_{>0}$, hence $(w+w')y\notin I^2$. Therefore,
  by Eisenstein criterion, $P(B)$ is irreducible in $(\RR[y])[x]=\RR[x,y]$.
\end{proof}


A rooted tree with branch lengths $B$ is a {\em branching tree} if and only if $B$ has no elementary vertices. 
In this section, we show that the polynomial $P(B,x,y)$ is a complete invariant for branching trees.
A branching tree $B$ is {\em proper} if and only if the degree of its root vertex is greater than one.

\begin{lemma}\label{lem:iso}
  For two proper branching trees $B$ and $B'$, if $P(B,x,y)=P(B',x,y)$, then $B$ and $B'$ are isomorphic. 
\end{lemma}

\begin{proof}
  We prove the lemma by strong induction on $n$, the degree on $x$ in the polynomials, which corresponds to the number
  of leaf vertices of the trees.

  If $n=1$, then the proper branching trees $B$ and $B'$ must consist only a single vertex.
  Thus, $P(B)=P(B')=x$ and $B$, $B'$ are isomorphic.

  Assume that the statement of the lemma is true for all pairs of proper branching trees with $n\le N$ leaf vertices.

  Now, suppose that $n=N+1$. 
  Let $B=\wedge_{i=1}^k w_iB_i$ and $B'=\wedge_{i=1}^l w'_iB'_i$ for $k,l>1$.
  We have $P(B)=\prod_{i=1}^k(w_iy+P(B_i))$
  and $P(B')=\prod_{i=1}^l(w'_iy+P(B'_i))$.
  As $P(B)=P(B')$, $\RR[x,y]$ is a unique factorization domain,
  and $w_iy+P(B_i)$ and $w'_iy+P(B'_i)$ are irreducible by Lemma \ref{lem:irred}, we have $k=l$ and, without loss of generality, $w_iy+P(B_i)=\alpha_i(w'_iy+P(B'_i))$ for
  some $\alpha_i\in\RR\setminus\{0\}$ 
  for all $i=1,2,\ldots,k$. 
  As both $P(B_i)$ and $P(B'_i)$ contain the term $x^m$ (with coefficient 1, where $m< N+1$) or all $i=1,2,\ldots,k$, we have that $\alpha_i=1$ for all $i=1,2,\ldots,k$.
  Since $B$ and $B'$ are proper branching trees, $B_i$ and $B'_i$ are also proper branching trees for all $i=1,2,\ldots,k$.
  So, none of the polynomials $P(B_i)$ or $P(B'_i)$ contains a term $ty$ for any $t\in\RR_{>0}$.
  Therefore, $w_i=w'_i$ and $P(B_i)=P(B'_i)$ for all $i=1,2,\ldots,k$.
  Now, as all trees $B_i$ and $B'_i$ have less than $n$ leaves, by the hypothesis, $B_i$ and $B'_i$ are isomorphic for all $i=1,2,\ldots,k$.
  Finally, as $B=\wedge_{i=1}^k w_iB_i$ and $B'=\wedge_{i=1}^k w'_iB'_i$ and $w_i=w'_i$ for all $i=1,2,\ldots,k$, the two proper branching trees $B$ and $B'$ are isomorphic.
\end{proof}

\begin{theorem}
    Two branching trees $B_1$ and $B_2$ are isomorphic if and only if the polynomials $P(B_1,x,y)$ and $P(B_2,x,y)$ are identical.
\end{theorem}

\begin{proof}
  The ``only if'' part directly follows Theorem~\ref{onlyif}.

  Notice that for a branching tree $B$, $P(B)$ contains a term $wy$ for some $w\in\RR_{>0}$ if and only if the root of $B$ has degree one, and $w$ is the branch length of the unique edge incident to the root of $B$.
  Suppose that $P(B_1)=P(B_2)$.
  Then, either both root vertices of $B_1$ and $B_2$ do not have degree one or both root vertices have degree one. 
  If both root vertices do not have degree one, then $B_1$ and $B_2$ are proper, and $B_1$ and $B_2$ are isomorphic by Lemma \ref{lem:iso}.
  If both root vertices of $B_1$ and $B_2$ have degree one, then $B_1=\wedge B'_1$ and $B_2=\wedge B'_2$, where $B'_1$ and $B'_2$ are two proper branching trees because $B_1$ and $B_2$ are branching trees which do not contain elementary vertices. 
  Furthermore, we have $P(B_1) = P(B'_1)+w_1y$ and $P(B_2) = P(B'_2)+w_2y$.
  Note that $B'_1$ and $B'_2$ are two proper branching trees, so neither $P(B'_1)$ nor $P(B'_2)$ has a term $ty$ for any $t\in\RR_{>0}$.
  Since $P(B_1)=P(B_2)$, $w_1=w_2$ and $P(B'_1)=P(B'_2)$.
  By Lemma \ref{lem:iso}, $B'_1$ and $B'_2$ are isomorphic, and so are $B_1$ and $B_2$.
\end{proof}

\section{Application to phylogenetic analysis}
\label{S:3}

Given a branching tree $B$, we can represent its polynomial $P(B,x,y)$ as a coefficient matrix, where $c^{(i,j)}$ is the coefficient of the monomial $x^iy^j$ in $P(B,x,y)$, where $i,j\geq 0$. 
We use the Canberra distance to measure the difference between two branching trees $B_1$ and $B_2$ \cite{liu2022metric}, which has been shown to be one of the most accurate metrics for tree polynomials \cite{liu2026comp}.
Let $c_1^{(i,j)}$ and $c_2^{(i,j)}$ denote the coefficients of the monomial $x^iy^j$ in $P(B_1,x,y)$ and $P(B_2,x,y)$, respectively.
The Canberra distance between the two branching trees $B_1$ and $B_2$ is defined by the following formulas (\ref{canbdist}) and (\ref{gamma}).
\begin{align}\label{canbdist}
    d(B_1,B_2) = \sum_{i,j\geq0} \gamma(c_1^{(i,j)}, c_2^{(i,j)})
\end{align}

\begin{align}\label{gamma}
\gamma(c_1, c_2) =
\begin{cases}
|c_1 - c_2|/(c_1 + c_2) & \text{if } c_1 \neq 0 \text{ or } c_2 \neq 0 \\
0 & \text{if } c_1 = 0 \text{ and } c_2 = 0
\end{cases}
\end{align}

We apply the polynomial-based Canberra distance to compare phylogenetic trees reconstructed from HIV-1 sequences in three contrasting epidemiological settings as in \cite{liu2022metric}: HIV-1 subtype B sampled from a concentrated epidemic among men who have sex with men in Seattle, USA (Wolf et al. \cite{wolf2017}); HIV-1 subtype C from a village-scale generalized epidemic in Mochudi, Botswana (Novitsky et al. \cite{novitsky2013}); and HIV-1 subtype C from a national survey of the generalized epidemic in South Africa (Hunt et al. \cite{hunt2012}).
These three trees are rooted by midpoint rooting and with branch lengths representing genetic divergence in units of expected substitutions per site \cite{liu2022metric}.
They capture HIV transmission at different spatial scales and across different epidemiological dynamics.
We compare only the underlying rooted trees with branch lengths of the three phylogenetic trees and do not differentiate labels of leaf vertices (taxa).
We use the three trees to demonstrate the capability of the polynomial in distinguishing trees with branch lengths arising from different epidemiological contexts.

The three HIV-1 phylogenetic trees have 1,653 (Wolf et al.), 782 (Novitsky et al.), and 968 leaf vertices (Hunt et al.), respectively. 
To make the trees comparable in size, we randomly subsample each phylogenetic tree to 500 leaf vertices, and draw 100 independent subsamples per source tree, yielding 300 trees in total.
We carry out two parallel analyses on this set: one in which the subsampled trees retain their original branch lengths (unscaled setting), and one in which each subsampled tree's branch lengths are rescaled so that the mean branch length equals one (normalized setting). 
For each of these two settings, we compute the polynomial $P(B,x,y)$ at every subsampled branching tree $B$, and then we compute the pairwise Canberra distance between each pair of the trees.
This results in two $300\times300$ Canberra distance matrices in the two settings, respectively.
We then cluster each distance matrix by $k$-medoids (partitioning around medoids or PAM) \cite{kaufman1990} with $k = 3$, and evaluate clustering accuracy against the known source-tree label.

The polynomial-based Canberra distance with $k$-medoids achieved 100.0\% accuracy (300/300) for subsampled trees with their original branch lengths, and 99.3\% (298/300) for subsampled trees with normalized branch lengths. 
We visualize the two $300\times300$ Canberra distance matrices using Multidimensional Scaling (MDS) \cite{coxcox2001}. See Figure~\ref{fig3}.
The clustering errors for branch-length-normalized trees consist of two Wolf subtrees misassigned to Novitsky. 
Both sitting at the Wolf/Novitsky interface in the MDS space. 
The Hunt cluster is perfectly recovered in both settings. 
For comparison, in \cite{liu2022metric}, the same three datasets were clustered using $k$-medoids with the lattice distance, and the reported accuracy is 72.7\% (misclassification rate 0.273).

\begin{figure}[ht!]
    \centering
    \includegraphics[width=1\linewidth]{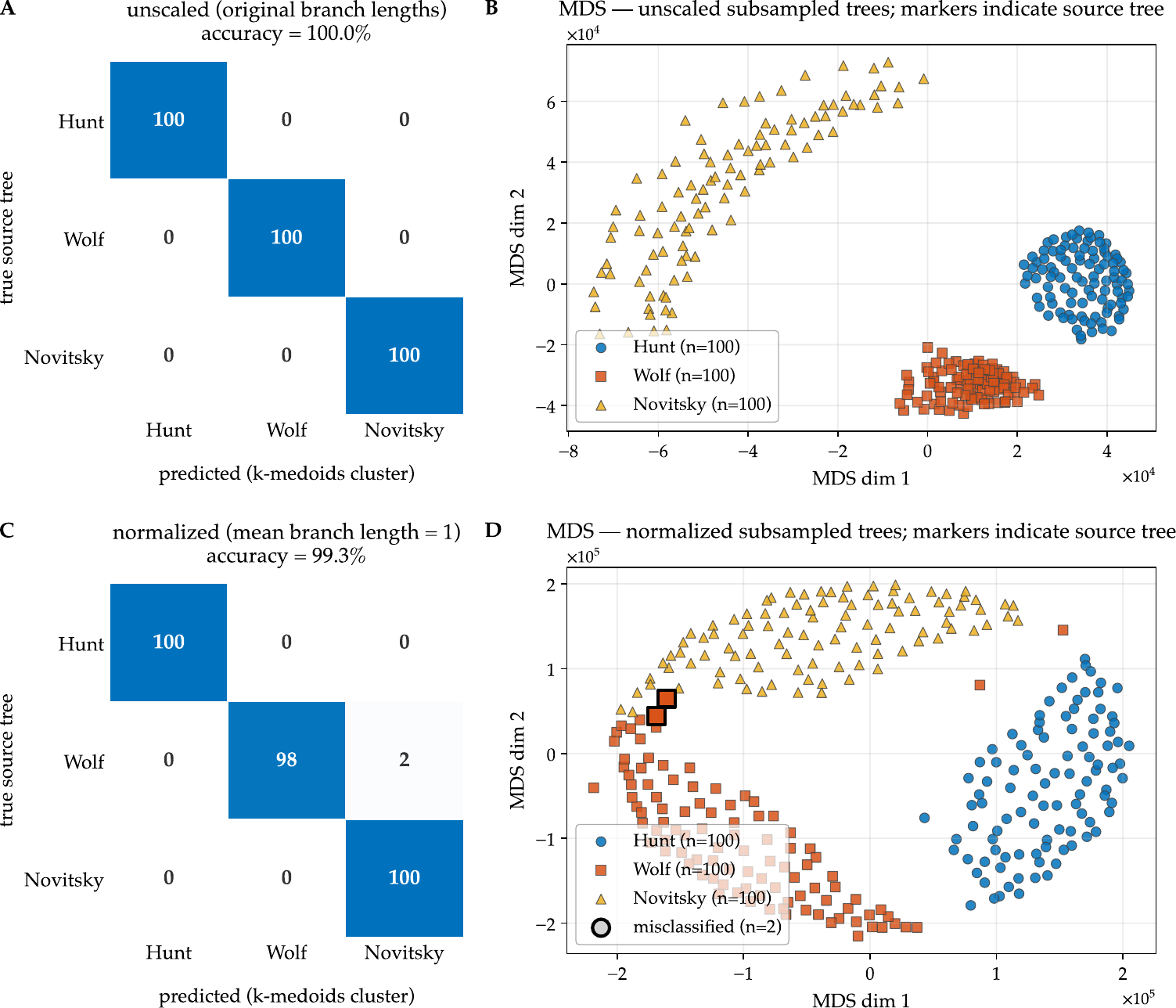}
    \caption{{\bf Clustering results of the HIV trees using the polynomial-based Canberra distance.} {A, C}: Confusion matrices from $k$-medoids clustering ($k=3$) using the polynomial-based Canberra distances between 300 subsampled trees, each with 500 leaf vertices (100 per source tree) under two settings: subsampled trees with original branch lengths (A; accuracy 100\%) and subsampled trees, each with mean branch lengths normalized to one (C; accuracy 99.3\%).
    {B, D}: MDS visualizations of the corresponding $300 \times 300$
    Canberra distance matrices; each marker represents a subsampled tree; marker shape and color denote the true source tree (Hunt: blue circles, Wolf: orange squares, Novitsky: yellow triangles). The two normalized-setting misclassifications are highlighted with black rings in panel D.}
    \label{fig3}
\end{figure}

We observe that the mean branch length of Novitsky subtrees is about 2.3 times larger than Wolf subtrees and about 1.7 times larger than Hunt subtrees.
These scale differences propagate through the polynomial computation, and the magnitudes of maximum coefficients in the polynomials are distinct for the three datasets, where Wolf subtrees have maximum coefficients in $10^5$, Hunt subtrees in $10^6$ and Novitsky subtrees in $10^{13}$.
The distinct mean branch lengths of the three datasets contribute to the 100.0\% accuracy of clustering for subsampled trees with original branch lengths.
When the mean branch lengths are normalized to one for all three datasets, the evolutionary signal captured by the polynomial encoding is still strong enough to recover 298 of 300 subtrees correctly.
This indicates that the three phylogenetic trees have distinct topologies and relative within-tree branch-length distributions.

\section*{Data availability}

Code and data for the experiment and analysis conducted in this paper are available at \url{https://github.com/pliumath/branching-tree-polynomial}.


\section*{Acknowledgments}

P.L. was supported by the start-up fund of the University of Rhode Island, and in part by the Rhode Island Institutional Development Award (IDeA) Network of Biomedical Research Excellence from the National Institute of General Medical Sciences of the National Institutes of Health under grant number P20GM103430. 
We would like to thank Olivier Gascuel and Manolo F. Perez for helpful discussions and suggestions regarding applications to phylogenetic analysis.

\nolinenumbers

\bibliographystyle{unsrtnat}
\bibliography{bibliography}

\end{document}